\documentclass{article}
\usepackage[utf8]{inputenc}
\usepackage{amsmath}
\usepackage{amsthm}
\usepackage{amssymb}
\usepackage{pgf}
\usepackage{tikz}
\usetikzlibrary{automata}
\usetikzlibrary{decorations,arrows}
\usetikzlibrary{decorations.pathmorphing}
\usetikzlibrary{positioning}
\usepackage[british]{babel}

\newcommand\definesymb[1]{%
\expandafter\newcommand\csname #1#1\endcsname{{\ensuremath{\mathbb{#1}}}}%
}

\definesymb{Z}
\definesymb{N}
\definesymb{R}

\newtheorem{theorem}{Theorem}
\newtheorem{defn}{Definition}[section]
\newtheorem{proposition}[defn]{Proposition}

\newtheorem{cor}[defn]{Corollary}

\title{Aperiodic Subshifts on Nilpotent and Polycyclic Groups}

\author{Emmanuel Jeandel\\
LORIA, UMR 7503 - Campus Scientifique, BP 239\\
54\,506 VANDOEUVRE-L\`ES-NANCY, FRANCE\\
\texttt{emmanuel.jeandel@loria.fr}}

\begin{document}

\maketitle

\begin{abstract}
We prove that every polycyclic group of nonlinear growth admits a strongly aperiodic SFT
and has an undecidable domino problem.
This answers a question of \cite{Carroll} and generalizes the result of \cite{BallierStein}.

\end{abstract}

Subshifts of finite type (SFT for short) in a group $G$ are colorings of the elements
of $G$ subject to local constraints.
They have been studied extensively for the free abelian group
$\mathbb{Z}$ \cite{LindMarcus} and on the free abelian group $\mathbb{Z}^2$, 
where they correspond to tilings of the discrete plane \cite{LindMulti}.

One of the main question about SFTs is the existence of strongly
aperiodic SFT, that is finding a finite set of local constraints so
that the only way to color the group $G$ is to do so aperiodically.
Strongly aperiodic SFT are known to exists for the group
$\mathbb{Z}^2$ \cite{BergerPhD} and the author \cite{SubGroups} has
provided examples on many groups of the form $\mathbb{Z} \times G$.

A related question is the decidability of the domino problem: Given an
SFT $X$ on a group $G$, decide if $X$ is empty.
This question is somewhat related to the existence of \emph{weakly}
aperiodic SFT, i.e. where no coloring is periodic along a subgroup of
finite index in $G$. The undecidability of the domino problem has been
established on $\mathbb{Z}^2$  \cite{BergerPhD}, on Baumslag-Solitar
groups \cite{AukA} and on any nilpotent group of nonlinear growth
\cite{BallierStein}.

In this article we introduce a few elementary techniques that show how
to exhibit strongly aperiodic SFT on any polycyclic group which is not
of linear growth (i.e. not virtually $\mathbb{Z}$), and we also prove
that any such group has an undecidable domino problem. The first
result answers a question of Carroll and Penland \cite{Carroll}, while
the second generalizes the result of Ballier and Stein from nilpotent
groups to polycyclic groups.

A full characterization of groups (or even finitely presented groups)
with strongly aperiodic SFT is still for now out of reach.
\section{Definitions}
We assume some familiarity with group theory and actions of groups.
See \cite{CicCoo} for a good reference on symbolic dynamics on groups.
All groups below are implicitely supposed to be finitely generated (f.g. for short).

Let $A$ be a finite set and $G$ a group.
We denote by $A^G$ the set of all functions from $G$ to $A$.
For $x \in A^G$, we write $x_g$ instead of $x(g)$ 
for the value of $x$ in $g$.

$G$ acts on $A^G$ by
\[
(g \cdot x)_h = x_{g^{-1}h}
\]

A \emph{pattern} is a partial function $P$ of $G$ to $A$ with finite
support. The support of $P$ will be denoted by $\mathrm{Supp}(P)$.

A \emph{subshift} of $A^G$ is a subset $X$ of $A^G$ which is
topologically closed (for the product topology on $A^G$) and invariant
under the action of $G$.

A subshift can also be defined in terms of forbidden patterns.
If $\cal P$ is a collection of patterns, the subshift defined by $P$ is 

\[
X_{\cal P} = \left\{ x\in A^G | \forall g \in G, \forall P \in {\cal P}
  \exists h \in \mathrm{Supp}(P), (g\cdot x)_h \not= P_h \right\}
\]
Every such set is a subshift, and every subshift can be defined this way.
If $X$ can be defined by a finite set $\cal P$, $X$ is said to be a
\emph{subshift of finite type}, or for short a SFT.

For a point $x \in X$, the stabilizer of $x$ is $Stab(x) = \{ g | g
  \cdot x =   x\}$

A subshift $X$ is \emph{strongly aperiodic} if 
it is nonempty and every point of $X$ has a finite stabilizer.
Some authors require that the stabilizer of each point is trivial
(rather than finite),
this will not make any difference in this article.

A f.g. group $G$ is said to have \emph{decidable domino problem} if there is an
algorithm that, given a description of a finite set of patterns $\cal P$, 
decides if $X_{\cal P}$ is empty.

In the remaining, we are interested in groups $G$ which admit 
strongly aperiodic SFTs or have an undecidable domino problem.

We now summarize previous theorems:
\begin{theorem}
\begin{itemize}
        \item $\mathbb{Z}$ does not admit strongly aperiodic SFT and has decidable word problem
        \item Free groups have decidable domino problem \cite{Kuske}
        \item The free abelian group $\mathbb{Z}^2$ \cite{BergerPhD}
		  has a strongly aperiodic SFT and an undecidable domino problem
		\item The free abelian group $\mathbb{Z}^3$ \cite{CulikKari}
		  has a strongly aperiodic SFT.
        \item f.g. nilpotent groups have an undecidable word problem unless they are virtually cyclic  \cite{BallierStein}
  \end{itemize}         
\end{theorem}  
\clearpage
We also give a few structural results:
\begin{theorem}
	\begin{itemize}
		\item Let $G,H$ be f. g. commensurable groups. Then $G$ admits a
          strongly aperiodic SFT (resp. has an undecidable domino problem)
          if only if $H$ does. \cite{Carroll}
		\item Let $G,H$ be \emph{finitely presented} groups that are
		  quasi-isometric. Then $G$ admits a
		  weakly aperiodic SFT (resp. has an undecidable domino problem)
		  if and only if $H$ does. \cite{Cohen2014}
        \item Finitely presented groups with a strongly aperiodic SFT
		  have decidable word problem \cite{SubGroups}.		  
	\end{itemize}		
\end{theorem}	
The first of these results will be used almost everywhere in what
follows.

\section{Building aperiodic SFTs on $G$ from aperiodic SFTs on subgroups and quotients of $G$}

\subsection{Building blocks}
We start by a few easy propositions that explain how to build SFT in
$G$ from a SFT in a subgroup $H$ or a SFT in $G/H$, provided in the second case
that $H$ is normal and finitely generated.

\begin{proposition}
	\label{prop:l1}
	Let $G$ be a f.g. group and $H$ a f.g. normal subgroup of $G$.
	Let $\phi: G \rightarrow G/H$ the corresponding morphism. 
    Let $X$ be a SFT on $G/H$ over the alphabet $A$.
    Let \[ Y = \{ y \in A^G | \exists x \in X \forall g \in G , y_g =
	  x_{\phi(g)} \} \]
	Then $Y$ is a SFT.
	
	Furthermore, for every $y \in Y$, there exists $x \in X$ s.t. $Stab(y) = H Stab(x)$.
\end{proposition}	
See e.g. \cite{Carroll} for a proof. The idea is to lift the forbidden
patterns defined on $G/H$ to forbidden patterns defined on $G$, and to
use additional forbidden patterns to force $y_g = y_{g'}$ whenever
$gg'^{-1} \in H$. This is possible as $H$ is supposed to be finitely
generated.

\begin{proposition}
		\label{prop:l2}
	Let $G$ be a f.g. group and $H$ a f.g. subgroup of $G$.
	
	Let $X$ be a SFT on $H$ over the alphabet $A$ defined
	by the set of forbidden	patterns $\cal P$.	
	Using the same forbidden patterns, we obtain a SFT on $G$
        that we call $Y$.
	
	Then $Y$ is nonempty iff $X$ is nonempty.
	More precisely, let $K$ be a left transversal of $H$ in $G$. 
	Then $y \in Y$ iff there exists points $(x^k)_{k \in K}$ in $X$  s.t. $y_{kh} = x^k_h$.

    In particular, for all $y \in Y$, there exists $x \in X$
	s.t. $Stab(y) \cap H \subseteq Stab(x)$
\end{proposition}
\begin{proof}
Write $G = KH$ for some transversal $K$. Each element of $g \in G$
can therefore be written in a unique way in the form $g = k h$.

Let $y \in Y$ and let $k \in K$. Define $(x^k)_h = y_{kh}$.
Let $h' \in H$ and $P$ a forbidden pattern of $\cal P$.
We will prove there exists $p$ s.t. $(h' \cdot x^k )_p \not= P_p$ which
proves that $x^k  \in X$.

By definition of $Y$, there exists $p$ in the support of $\cal P$ s.t.
$(h' k^{-1} \cdot y)_p \not= P_p$.
Therefore $y_{kh'^{-1} p} \not= P_p$. But $(h' \cdot x^k)_p =
(x^k)_{h'^{-1} p} = y_{kh'^{-1}p}$, the result is proven.

Conversely, take some points $(x^k)_{k \in K}$ in $X$ and define
$y_{kh} = x_h^k$.
Let $g \in G$. Write $g^{-1} = kh^{-1}$ for some $k$ and $h$.
Let $P\in \cal P$.
As $x^k \in X$, there exists $p \in H$ s.t. $(h \cdot x^k)_p \not= P_p$.
But then 
$(g \cdot y)_p = y_{g^{-1}p} = y_{kh^{-1}p} = x^k_{h^{-1}p} = (h \cdot x^k )_p \neq P_p$.
Therefore $y \in Y$.	
\end{proof}	

\subsection{Applications}
We now start with the first proposition that gives a natural way to
prove that a group $G$ has a strongly aperiodic SFT.

\begin{proposition}
  \label{main:prop}
	Let $G$ be a f.g. group. Suppose that $G$ contains two f.g. groups $H_1, H_2$ s.t.
	\begin{itemize}
	    \item $H_1 \subseteq H_2$.
		\item $H_1$ is normal and $G/H_1$ admits a strongly aperiodic SFT $X_1$
		\item $H_2$ admits a strongly aperiodic SFT $X_2$
	\end{itemize}		
   Then $G$ admits a strongly aperiodic SFT $Y$.
   Furthermore, if each point in $X_1$ and $X_2$ have trivial
   stabilizers, then every point of $Y$ has a trivial stabilizer.
\end{proposition}
It is a good but nontrivial exercise to show, under the conditions of the
proposition, that if $H_1$ and $H_2$ have a decidable word problem,
then $G$ does.
\begin{proof}
	Use the two previous propositions to build $Y_1$ and $Y_2$ and consider $Y = Y_1 \times Y_2$.
	
	Let $y = (y_1, y_2) \in Y$ and consider its stabilizer $K = Stab(y)$.
	By definition of $Y_1$, $K \subseteq H_1 F$ for some finite set $F$.
	By definition of $Y_2$, $K \cap H_2$ is finite. In particular $K \cap H_1$ is finite.
	
	This implies that $K$ is finite. Indeed, for each $f \in F$, $K
	\cap fH_1 = K \cap H_1f$ is finite: If $x,y \in K \cap fH_1$, then $xy^{-1} \in K \cap H_1$.
			
    Furthermore if $F$ is trivial and $K \cap H_2$ is trivial, then	$K$ is trivial.
\end{proof}

\begin{cor}
	\label{cor:l1}	
Let $G$ be a f.g. group.
If $G$ contains a f.g. normal subgroup $H_1$ s.t. $G/H_1$ admits a
strongly aperiodic SFT and $H_1$ admits a strongly aperiodic SFT, then
$G$ admits a strongty aperiodic SFT.

In particular, if $G_1$ and $G_2$ are f.g. groups that admit strongly
aperiodic SFTs, then $G_1 \times G_2$ does.
\end{cor}
Therefore groups with strongly aperiodic SFT are closed under direct
product. This seems quite natural but does not seem to have been known.

The second proposition of the previous section explain
how an aperiodic SFT on $H \subseteq G$ give rise to a SFT on $G$.
Usually this SFT will not be aperiodic. There are however a few
conditions in which it will.
Here is an obvious one, for which we do not have any interesting application
\begin{proposition}
Let $H$ be a subgroup of $G$ that admits an aperiodic SFT $X$.
Suppose that every element of $G$ is conjugate to some element of $H$
($H$ is said to be \emph{conjugately-dense}).
Then $G$ admits an aperiodic SFT. Actually, $X$, seen as an aperiodic
SFT on $G$, is aperiodic.
\end{proposition}

\section{Aperiodic SFTs on polycyclic and nilpotent groups}
Using the previous proposition, we will be able to prove that
nontrivial nilpotent and polycyclic groups have strongly aperiodic SFT.
We start with nilpotent groups.

First, a few definitions, see \cite[Chapter 1]{Segal} for details.

Let $G$ be a  group. The \emph{Hirsch number} $h(G)$ of $G$ is the number of infinite
factors in a series with cyclic or finite factors.
It is defined for nilpotent groups, and more generally for polycyclic groups.
A nilpotent (or a polycyclic) group of Hirsch number $0$ is a finite group.
A nilpotent (or a polycyclic) group of Hirsch number $1$ is a finite-by-cyclic-by-finite group, hence virtually $\mathbb{Z}$.

The only thing we will need about nilpotent  groups is that (a) they
have a non trivial center (b)  quotients and
subgroups of nilpotent groups are nilpotent and finitely generated (c)
$h(G) = h(G/H)+h(H)$ (which make sense due to the previous point) (d)
every nilpotent group contains a torsion-free nilpotent group of
finite index.

Note that the rank and the Hirsch number coincide for free abelian
groups, i.e. $\mathbb{Z}^n$ is of Hirsch number $n$.

\begin{theorem}
	Let $G$ be a finitely generated group of nonlinear polynomial growth.
	
	Then $G$ admits a strongly aperiodic SFT.
\end{theorem}
\begin{proof}
Such groups are exactly the f.g. virtually nilpotent groups of Hirsch
number different from $1$.
The proof will use repeatedly that if $G$ is a f.g. group, and $H$ a
subgroup of finite index, then $G$ has a strongly aperiodic SFT iff
$H$ does, see \cite{Carroll} for details.
In particular it is sufficient to prove the theorem for
f.g.  nilpotent groups to obtain the results for f.g. virtually nilpotent groups.

The result is clearly true for Hirsch number $0$. We now prove
the theorem by induction on the Hirsch number $h(G) \geq 2$.

We start with Hirsch number 2. A nilpotent group of Hirsch number $2$ is virtually $\mathbb{Z}^2$, hence has a strongly aperiodic SFT.

Now let $G$ be a nilpotent group of Hirsch number $n > 2$.
Let $G_1$ be a torsion-free nilpotent subgroup of finite index of $G$.
It is enough to prove the theorem for $G_1$ to obtain the theorem for
$G$, therefore we will suppose that $G_1 = G$.

$G$ has a non trivial center and torsion-free, hence contains a copy
of $H_1 = \mathbb{Z}$ in its center.
$h(G/H_1) = h(G) - h(H_1) \geq 2$ therefore by induction $G/H_1$ has a strongly aperiodic SFT.

Furthermore, $G/H_1$ is a nilpotent group, hence contains a
torsion-free nilpotent group of finite index, hence contains a
torsion-free element $x$. Let $y$ be a representative of $x$ in $G$.
Let $H_2$ be the group generated by $H_1$ and $y$.
AS $H_1$ is normal in $H_2$, $h(H_2) = h(H_2/H_1) + h(H_1) = 2$
Therefore $H_2$ is by construction of Hirsch length $2$, hence has a strongly aperiodic SFT.

We have produced our required groups $H_1$ and $H_2$ and we can apply
Proposition~\ref{main:prop} to finish the induction.
\end{proof}	

We now proceed to the proof for a polycyclic group. The proof is
roughly similar, except there is an additional case to work out, which
is the case of groups $G$ that admit an exact sequence $1 \rightarrow
\mathbb{Z}^n \rightarrow G \rightarrow \mathbb{Z} \rightarrow 1$.
It is not obvious how to obtain strongly aperiodic SFTs for these
groups. Fortunately:
\begin{theorem}[\cite{BarSab}]
  Let $G$ be a group that admits an exact sequence
  \[1 \rightarrow \mathbb{Z}^n \rightarrow G \rightarrow \mathbb{Z}
  \rightarrow 1\] with $n \geq 2$.

  Then $G$ admits a strongly aperiodic SFT.
\end{theorem}
This result is highly nontrivial and $\mathbb{Z}$ can actually be
replaced by any f.g. group with decidable word problem. We prefer
citing the result in this form as it is likely it admits a simpler
proof in this particular case.

Polycyclic groups are groups which admit a series with cyclic factors.
All relevant properties of nilpotent groups stated above are still true when
nilpotent is replaced by polycyclic: polycyclic groups are always
finitely generated (actually finitely presented), quotients and
subgroups of polycyclic groups are polycyclic, 
every polycyclic group contains a torsion-free polycyclic group of
finite index. Moreover polycyclic
groups admit a nontrivial normal free abelian subgroup.

\begin{theorem}
	Let $G$ be a virtually polycyclic group of Hirsch number different from $1$.
	
	Then $G$ admits a strongly aperiodic SFT.
\end{theorem}
\begin{proof}

We start with Hirsch number 2. In this case, the polycyclic group $G$ is virtually $\mathbb{Z}^2$, hence has a strongly aperiodic SFT.

For the induction, consider $G$ a polycyclic group of Hirsch number $n> 2$.
$G$ contains a torsion-free polycyclic group of
finite index so we may suppose as before that $G$ is torsion free.

$G$ contains an nontrivial normal free abelian subgroup $H_1 = \mathbb{Z}^k$.

There are four cases:
\begin{itemize}
 \item $k = n$. In this case, $h(G/H_1) = h(G) - h(H_1) = 0$, therefore $G$
  is virtually $\mathbb{Z}^n$ and admits a strongly aperiodic SFT.
\item $k = 1$. We use the induction hypothesis on $G/H_1$ of Hirsch number
  at least $n-1$. By taking any element of infinite order in $G/H_1$, we
  obtain, as in the nilpotent case, a group $H_2 \supset H_1$ of Hirsch
  number $2$ and we apply Proposition~\ref{main:prop}.
\item $1 < k < n-1$. We use the induction hypothesis on both $G/H_1$ and
  $H_2 = H_1$ and conclude by Proposition~\ref{main:prop}.
\item $k = n-1$.
  In this case $G/H_1$ is of Hirsch number $1$, hence virtually
  $\mathbb{Z}$.
  
  By taking a finite index subgroup $G_1$ of $G$ we can suppose wlog that
  $G/H_1$ is exactly $\mathbb{Z}$. We can then apply the theorem of \cite{BarSab}.
\end{itemize}  
\end{proof}	
Note that the full extent of the theorem in \cite{BarSab} covers all cases
except the case $k \leq 1$. The reason we stated it only for the case $k =
n-1$ is that we believe there is an easier proof in this case.

\section{The domino problem}
We now prove that the domino problem is undecidable for all virtually
polycyclic groups which are not virtually cyclic (i.e. which are of
Hirsch number greater than $2$).
\begin{proposition}
	Let $G$ be a virtually polyclic group of Hirsch number greater than $2$.
	Then $G$ admits a f.g. subgroup $H$ that factors onto $\mathbb{Z}^2$.		
\end{proposition}	
\begin{proof}
	By induction on the Hirsch number, it is sufficient to prove the
	result for polycyclic groups.	
	Every polycyclic group of Hirsch number $2$ is virtually
	$\mathbb{Z}^2$, there is nothing to prove.
	
	Let $G$ be a polycyclic group of Hirsch number at least 3 and let $N$ be a
	nontrivial normal free abelian subgroup of $G$.
	If $N$ is of rank greater than $2$, $G$ contains a subgroup
	isomorphic to $\mathbb{Z}^2$, there is nothing to prove.
	
	Otherwise, $G/N$ is a polycyclic group of Hirsch number at least
	2, and admits by induction a subgroup $H$ that factors onto $\mathbb{Z}^2$.
	Therefore $G$ admits a subgroup (the preimage of $H$) that factors
	onto $\mathbb{Z}^2$.
\end{proof}	

\clearpage
\begin{cor}
	Every virtually polycyclic group which is not virtually cyclic has
	an undecidable domino problem.
\end{cor}
\begin{proof}
  We use the previous proposition to obtain $H$.
	As $G$ is polycyclic, $H$ is polycyclic.
	Then $H/N = \mathbb{Z}^2$ for a normal subgroup $N$ of $H$, which
	is finitely generated as $G$ (therefore $H$) is polycyclic.
	
	Therefore $H$ has an undecidable domino problem: Given a SFT $X$
	on $H/N$, we may obtain (constructively) a SFT $Y$ on $H$ s.t. $X$
	is empty iff $Y$ is empty by Proposition \ref{prop:l1}.
	
	Therefore $G$ has an undecidable domino problem: Given a SFT $X$
	on $H$, we may obtain (constructively) a SFT $Y$ on $G$ s.t. $X$
	is empty iff $Y$ is empty  by Proposition \ref{prop:l2}.
\end{proof}

\section{Conclusion}

Proposition \ref{prop:l2} gives a way
to transform a SFT $X$ on $H$ to a SFT $Y$ on $G \supseteq H$. In many
cases, $Y$ will not be strongly aperiodic. However, surprisingly, it
is the case  in some situations. We already have used this fact in the
proof of our main theorem. Another example is that $\mathbb{Z}[1/2]
\times \mathbb{Z}[1/2]$ has a strongly aperiodic SFT $Y$ , which is
basically obtained from any strongly aperiodic SFT $X$ on $\mathbb{Z}^2$.
The idea is that if $g \in \mathbb{Z}[1/2]
\times \mathbb{Z}[1/2]$ is nontrivial, then some power of $g$ is
in
$\mathbb{Z} \times \mathbb{Z}$, therefore no nontrivial element of $g$
can be stabilizer of some point of $Y$.

Of course $ \mathbb{Z}[1/2] \times \mathbb{Z}[1/2]$ is an infinitely
generated group, and the study of SFT is primarily interesting in finitely
generated groups. However this observation might be useful for example
to produce a strongly aperiodic SFT in some solvable groups, e.g.
$BS(1,2) \times BS(1,2)$ ($BS(1,2)$ is basically a semidirect product
of $\mathbb{Z}[1/2]$ and $\mathbb{Z}$).

An obvious natural generalization of the main theorem would be to deal with
solvable groups rather than
polycyclic groups. The difficulty is that solvable groups that are not
polycyclic always have abelian  subgroups which are not finitely
generated. This adds additional cases that we do not know how to treat.
Of particular interest is the lamplighter group, or solvable
Baumslag-Solitar groups (for which only weakly aperiodic SFTs are known).
Furthermore, it is not true that every solvable group has
a strongly aperiodic SFT as there exist solvable finitely presented
group with   an undecidable word problem \cite{Baumslag}.

Another possible generalization which is  promising is Noetherian
groups, which are groups where every subgroup is finitely generated.
The wilder examples of Noetherian groups, Tarski monsters, do admit
strongly aperiodic SFT $X$ \cite{SubGroups} (where points in $X$ might
have finite, non trivial stabilizer), so it is quite possible that
all Noetherian groups (which are not virtually cyclic) do admit
strongly aperiodic SFT.

\appendix
\section{Automorphism-free SFT}

We finish  with a new definition of aperiodicity 
\begin{defn}
Let $G$ be a group. Let $Aut(G)$ be the group of automorphisms of $G$.
For $\phi \in Aut(G)$, let $\phi(x)$ be the point defined by $\phi(x)_g = x_{\phi(g)}$.

For a subshift $X$ and a point $x$, let
$Div(x,X) = \{ \phi \in Aut(G), \phi(x) \in X\}$

We say that a nonempty subshift  $X$ is automorphism-free if $Div(x,X)$ is
trivial for all $x \in X$.

\end{defn}

The interest of automorphism-free SFT is seen by the following remarks:
\begin{proposition}
Let $G = \mathbb{Z}^n$ with $n > 1$.
Then an automorphism-free SFT $X$ of $G$ is strongly aperiodic.
\end{proposition}
\begin{proof}
	In this proof, we will see $\mathbb{Z}^n$ as a $\mathbb{Z}$-module
	and will write the law group additively and not multiplicatively.

Let $x \in X$.
Let $u \in \mathbb{Z}^n$. 
Let $v$ be any nonzero element of $\mathbb{Z}^n$ that is orthogonal to $u$, that is $uv^T = 0$.

Define $\phi(g) = g + (gv^T) u $.
Then $\phi$ is an automorphism of $\mathbb{Z}^n$, the inverse being given
by $\psi(g) = g - (gv^T) u$. (In a basis with
some basis vectors collinear to $u$ and $v$, $\phi$ would be a identity matrix with
one other nonzero coefficient).

Suppose that $u$ is a vector of periodicity for $x$, i.e.  $u\cdot x =
x$, therefore $x_{g - u} = x_g$ for all $g \in \mathbb{Z}^n$

However, for all $g \in \mathbb{Z}^n$, $g v^T$ is an integer, therefore
$\phi(g) = g + k u$ for some $k$ depending on $g$.
Therefore $x_{\phi(g)} = x_g$.

In particular $\phi(x) = x$. Therefore $\phi(x) \in X$.
Therefore $\phi$ is trivial, that is $u = 0$.
\end{proof}

And the fact that automorphism-free SFT do exist:
\begin{proposition}
	Let $n > 1$. Then $\mathbb{Z}^n$ admits a automorphism-free SFT.
\end{proposition}	
\begin{proof}
Let $\{e_1, e_2\dots e_n\}$ be the canonical base of $\mathbb{Z}^n$.
Let $V_i$ be the module generated by all vectors $\{ e_j, j \not= i\}$.

We will first build, for all $i$, a SFT $X_i$ s.t. if $x \in X_i$ and
$\phi$ is an automorphism of $\mathbb{Z}^n$ s.t. $\phi(x) \in X_i$, then $\phi(V_i) = V_i$.
It is sufficient to do it for $i = 1$.

To do this, we will start from a SFT $Y$ on $A^{\mathbb{Z}^2}$ built by Kari \cite{Kari14}.
On this SFT, a mapping $\pi: A \mapsto \{0,1\}$ can be defined s.t.
for any configuration $x \in Y$, 
every column of $\pi(x)$ (i.e. in direction $e_2$) is monochromatic,
and every row contains a sturmian word of irrational slope.
This implies in particular that every line of the form $(k p e_1)_{k \in
  \mathbb{Z}}$, with $p \neq 0$,  cannot be monochromatic.

We extend this SFT $Y$ to a SFT $X_1$ on $\mathbb{Z}^n$ by imposing every
configuration to be periodic in direction $e_3 \dots e_n$.
Doing this, we have built a (nonempty) SFT with the following property: for
every point $x \in X_1$:
\begin{itemize}
	\item $\pi(x_{p_1 e_1 + p_2 e_2 + \dots p_n e_n}) = \pi(x_{p_1 e_1})$
	\item $(\pi(x_{p_1 k e_1}))_{k \in\mathbb{Z}}$ is not
	  monochromatic unless $p_1 = 0$.
\end{itemize}

Now, let $x \in X_1$ and $\phi$ an automorphism of $\mathbb{Z}^n$ s.t. $y = \phi(x) \in X_1$.
Let $\phi(e_2) = p_1 e_1 + p_2 e_2 + \dots p_n e_n$.
As $y \in X_1$, $(\pi(y_{ke_2}))_{k \in \mathbb{Z}}$ must be monochromatic.
But $\pi(y_{ke_2}) = \pi(\phi(x)_{ke_2}) = \pi(x_{\phi(ke_2)}) =
\pi(x_{kp_1e_1})$. Therefore $p_1 = 0$.
Doing the same with $e_3 \dots e_n$, we have proven that if $\phi(x)
\in X_1$, then $\phi(V_1) = V_1$.

We build in the same way SFTs $X_2, X_3 \dots X_n$ with similar
properties and we take $X = X_1 \times X_2 \dots \times X_n$.

By the previous discussion, if $x \in X$ and $\phi(x) \in X$ then
$\phi(V_i) = V_i$ for all $i$.
This implies in particular $\phi(e_i) = \pm e_i$.

To finish, let $Z$ be the SFT over the alphabet $\{0,1,2\}^n$ that
consists in the point $z$ defined by $z_{p_1e_1 + \dots p_n e_n} =
(p_1 \mod 3, p_2 \mod 3, \dots p_n \mod 3)$ and its $3^n - 1$ other translates.

It is easy to see that if $z \in Z$ and $\phi(z) \in Z$ then it is not
possible to have $\phi(e_i) = -e_i$ for some $i$.

Therefore $X \times Z$ is automorphism-free.
\end{proof}

\end{document}